\documentclass[letterpaper, 10 pt, conference]{ieeeconf} 
\usepackage{cite}
\usepackage[pdftex]{graphicx}
\usepackage[cmex10]{amsmath}
\usepackage{amssymb,graphicx,charter, latexsym}
\usepackage{amsfonts}
\usepackage{bbm}
\usepackage{amsfonts}
\usepackage{bbm}
\usepackage{algorithm}
\usepackage{algorithmic}
\usepackage{tikz} 
\usepackage{comment}
\usepackage{amsfonts}
\usepackage{verbatim}
\usepackage{amssymb}
\newtheorem{lemma}{Lemma}
\usepackage{comment}
\newtheorem{theorem}{Theorem}
\newtheorem{condition}{Condition}
\newtheorem{definition}{Definition}

\IEEEoverridecommandlockouts                              

\title{Risk-Sensitive Optimal Control of Queues}
\author{Rahul~Singh, Xueying Guo, and Eytan Modiano
\thanks{Rahul Singh and Eytan Modiano are with the Laboratory of Information and Decision Systems (LIDS), Massachusetts Institute of Technology, Cambridge, MA 02139, USA.
        {\tt\small rsingh12@mit.edu, modiano@mit.edu}}%
\thanks{ Xueying Guo is postdoctoral researcher at Computer Science Department, University of California, Davis
		{\tt\small xyguo@ucdavis.edu}}
		}
\begin{document}
\maketitle
\begin{abstract}
We consider the problem of designing risk-sensitive optimal control policies for scheduling packet transmissions in a stochastic wireless network. A single client is connected to an access point (AP) through a wireless channel. Packet transmission incurs a cost $C$, while packet delivery yields a reward of $R$ units. The client maintains a finite buffer of size $B$, and a penalty of $L$ units is imposed upon packet loss which occurs due to finite queueing buffer.   

We show that the risk-sensitive optimal control policy for such a simple set-up is of threshold type, i.e., it is optimal to carry out packet transmissions only when $Q(t)$, i.e., the queue length at time $t$ exceeds a certain threshold $\tau$. It is also shown that the value of threshold $\tau$ increases upon increasing the cost per unit packet transmission $C$. Furthermore, it is also shown that a threshold policy with threshold equal to $\tau$ is optimal for a set of problems in which cost $C$ lies within an interval $[C_l,C_u]$. Equations that need to be solved in order to obtain $C_l,C_u$ are also provided. 
\end{abstract}
\section{Introduction}
In this work we consider the risk-sensitive optimal control of a one-hop \emph{stochastic} wireless network that comprises of a single client. Networked control systems are becoming increasingly susceptible to attacks~\cite{cardenas2008research}, and tools such as risk-sensitive and robust control can play an important role in securing these systems. Employement of a risk-sensitive control policy can serve as a mechanism to protect the network against attacks such as denial-of-service attacks.

Consider a denial-of-service attack carried out by a stochastic adversary that expends power in order to jam the communication channel between the client and the AP. Utilizing a risk-sensitive network control policy will make the closed-loop system more robust to the errors in the modelling assumptions made on the adversarial attack. The risk-sensitive optimal control policy hedges against the uncertainty by placing a greater emphasis on system trajectories that incur higher operation costs. If $c(t),t=1,2,\ldots,T$ denotes the instantaneous cost incurred during time $t$, then the risk-sensitive cost with risk-sensitivity parameter $\gamma>0$ incurred during time period $T$ is given by
\begin{align*}
\mathbb{E}e^{\gamma \sum_{t=1}^{T}c(t)},
\end{align*}
where expectation is taken with respect to the arrival process, the control policy used for scheduling packets, and the departure process. 
In the large-risk limit, i.e., $\gamma\to\infty$, the risk sensitive cost approaches the minimax cost objective, see~\cite{coraluppi1999risk}. Since the minimax objective seeks to minimize the system cost for the worst case scenario, a risk sensitive controller designed with risk parameter $\gamma$ set to a large value, has a good performance in case the system dynamics are ``adversarial" in nature. The framework provides flexibity by allowing the network operator to choose between the two competing objectives of having low risk-neutral cost, and that of making the system safe against attacks by tuning the risk-sensitivity parameter $\gamma$. Risk-sensitive control theory builds upon the ideas of Dynamic games and robust control~\cite{James1992,Jacobson1973,bensoussan1985optimal,james1994risk} and allows the system operator to generate control actions that reflect his confidence about the uncertainty in the model of the attack. It also generlizes the risk neutral approach towards dynamic optimization~\cite{Marcus1997}. Risk-sensitive control approach provides a link between the stochastic and deterministic approaches to model system uncertainty~\cite{Fleming1997,fleming1995risk}.


Risk-sensitive optimization places emphasis on higher order moments of the system cost~\cite{kumar1981optimal}, and thus risk-sensitive optimal control reduces undesirable stochastic variations in the system performance. This is highly desirable for network control systems in which the control loop is closed over stochastic communication networks~\cite{rahul,rahul1,guosingh,singh_stolyar_info,rs,singh2017optimal,singh2016throughput}. Risk sensitive system cost takes into account higher order moments of the (random) cost as well, as opposed to the risk neutral cost objective which only inlcudes the mean cost. Since risk sensitive cost objective penalizes higher order moments, it allows for designing a finer controller for the cost of interest. 

We discuss past works dealing with results on risk-sensitive control, and their applications in security of network control systems in Section~\ref{sec:pw}. The set-up involving single client being served by an access point is introduced in Section~\ref{sec:smsc}. We derive the structure of the optimal policy for single client scheduling problem in Section~\ref{sec:struc}. Section~\ref{sec:opt} derives the set of transmission costs for which threshold policy with threshold equal to $\tau$ is optimal. 
 Section~\ref{conclu} discusses directions for future research, and also summarizes the key results of this paper. 
 \begin{figure}
 \includegraphics[scale=.38]{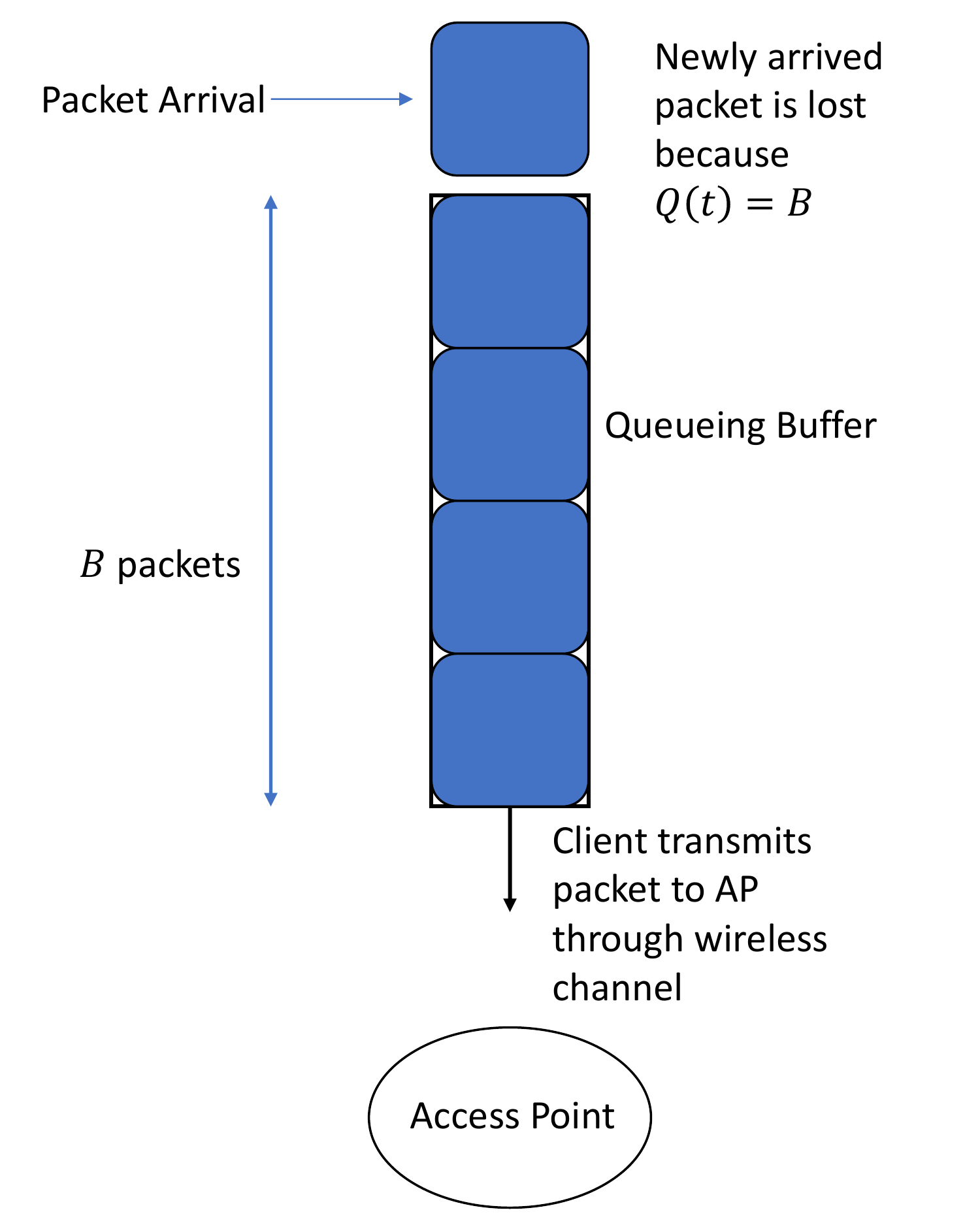}
 \caption{A single client uses the stochastic wireles channel for carrying out packet transmissions. Size of the queueing buffer is equal to $B$ packets, and if a packet arrives at a time $t$ when the queue length $Q(t)=B$, then it is lost.}
 \label{f1}
 \end{figure}
 \section{Past Works}\label{sec:pw}
The work~\cite{Howard1972} is one of the first to consider the problem of dynamic optimization of risk-sensitive cost within the Markov Decison Process (MDP) framework. For linear systems driven by Gaussian noise and quadratic one-step cost,~\cite{Jacobson1973} shows that the risk-sensitive controller depends upon the variance of noise, which is unlike the case of risk-sensitive LQG control. For a detailed treatment of risk-sensitive control of LQG systems, see~\cite{whittle1990risk}. Results concerned with risk-sensitive control of finite-state discrete-time controlled Markov chains can be found in~\cite{coraluppi1999risk}, while~\cite{Marcus1997} provides an overview of key results in risk-sensitive control.  

In recent years, the problem of desgning protocols and control policies for networked systems and crucial infrastructure such as sensor networks, electric power grids etc. has gained much attention~\cite{bompard2009risk,ericsson2007toward}. 
\cite{amin2009safe} considers the design of risk sensitive controller for a networked control system that is susceptible to denial-of-service attacks. The dynamical system of interest is assumed to be linear. \cite{befekadu2011risk} studies risk-sensitive control in the context of denial of service attacks in network.~\cite{singh2015index} derives scheduling policies that perform a mean versus variance trade-off with respect to packet interdelivery times.

Existing literature on stochastic control of queueing networks has mainly focused on risk-neutral cost objective. Works such as~\cite{1102957,buyukkoc1985c,lin1984optimal,stidham1989monotonic,george2001dynamic,harisondiscount} have derived optimal control policy and its structure under various assumptions regarding the stochastic queueing network. However there seems to be a gap with regards to the design of risk-sensitive control in the context of queueing networks.  
\section{Single Client Scheduling Problem}\label{sec:smsc}
We begin by describing the risk-sensitive queue control problem involving a single client being served by an unreliable channel. 

\emph{Continuous Time Model} 
The system begins operation at time $t=0$, and the packet arrivals to the client are governed by a Poisson process with rate $\lambda$. Let $Q(t),t\geq 0$ denote the queue length of the buffer at time $t$. If the client decides to carry out packet transmission at time $t$, then the time taken to complete packet transmission is exponentially distriibuted with mean $1/\mu$. During the time of packet transmission, cost is incurred at the rate of $C$ units per unit time. The cost $C$ models the amount of power utilized for packet transmission through the wireless medium. A reward of 
$R$ units is generated upon a successful packet transmission, or equivalently the delivered packet is counted towards the network throughput~\cite{tassi1}. The client maintains a queueing buffer of size $B$ packets. A packet loss occurs at time $t$ if a packet arrives and $Q(t)=B$, i.e. the queue buffer is full. The system is penalized $L$ units upon a packet loss.  

\emph{Equivalent Discrete-Time Model} The continuous-time discrete space Markov process described above can be converted into an equivalent discrete-time Markov chain by sampling the embedded Markov chain at time epochs when a packet arrival or departure occurs. Such technique is commonly utilized in the analysis of queueing systems, see ~\cite{1102957} or Ch:10 of~\cite{sennott2009stochastic} for a detailed discussion. We now describe the discrete-time system in detail.

Let $Q(t)$ denote the queue length of the buffer at time $t$. The queue length $Q(t)$ of the client evolves over discrete time-slots $t=1,2,\ldots$. At each time $t=1,2,\ldots$, the client can choose to either attempt packet transmission, i.e., $U(t)=1$, or stay idle $U(t)=0$. If $Q(t)>0$ and the client attempts a packet transmission at time $t$, then the queue length at time $t+1$ is equal to $Q(t)-1$ with a probability $p$, while it is equal to $(Q(t)+1)\wedge B$ with a probability $1-p$.
The quantity $p$ is equal to the probability with which the packet transmission completes before a new packet arrives in the original continuous-time model and is equal to $\mu/(\lambda+\mu)$. The client is charged $C>0$ units for attempting to transmit packet, and is provided a reward of $R>0$ units upon successful packet delivery. 

If at time $t$ either the client decides to not carry out packet transmission, or if $Q(t)=0$, then the queue length $Q(t+1)$ is equal to $(Q(t)+1)\vee B$ with probability $1$. If an arriving packet at time $t$ finds the queueing buffer full, i.e., $Q(t)=B$, then the packet is lost and the system is penalized $L>0$ units. Figure~\ref{f1} depicts the wireless network of interest.
A history dependent scheduling policy $\pi$, for each time $t=1,2,\ldots$ maps the history of the system until time $t$ to an action $U(t)\in\{0,1\}$. A Markov policy $\pi$ maps the queue length $Q(t)$ at time $t$ to a decision $U(t)\in \{0,1\}$. The infinite-horizon risk-sensitive cost incurred by the system is equal to
\begin{align}\label{rsmdp}
\min \limsup_{T\to\infty}\frac{1}{\gamma T}\log \mathbb{E} \left\{\exp\gamma\left(\sum_{t=1}^{T} CU(t) - R(t)+L(t)\right) \right\},
\end{align}
where the random process $R(t)$ assumes the value $R$ if a packet is deliverd at time $t$, while is $0$ otherwise, and the process $L(t)$ assumes the value $L$ if a packet is lost at time $t$, and is $0$ otherwise. The parameter $\gamma>0$ controls the sensitivity of the client towards the risk, and is called risk-sensitivity parameter~\cite{whittle1990risk,kumar1981optimal}. If for any Markov policy $\pi$, the process $Q(t)$ is irreducible and aperiodic, the $\limsup$ in the above definition can be replaced by $\lim$~\cite{borkar2002q}. We briefly discuss the existing results on infinite horizon risk-sensitive control for finite-state Markov chains.

\emph{Results on Infinite Horizon Risk-Sensitive Control}
Let us denote by $\pi^\star$ the policy that is optimal for the risk-sensitive MDP~\eqref{rsmdp}.
It can be shown that (\cite{Rojas1998,Marcus1997,borkar2002q}) there exists a value function $V: [0,B]\mapsto \mathbb{R}$, and a scalar $\alpha$, such that
\begin{align}\label{fp}
\alpha V(i) = \min_{u\in \{0,1\}} \sum_{j\in [0,B]} e^{Co(i,j,u)}p(j|i,u)V(j), i\in[0,B]
\end{align}
where $p(j|i,u)$ is the transition probability associated with state $i$ to state $j$ under the application of control action $u$, and $Co(i,j,u)$ is the one-step cost associated with the state-action pair $(i,a)$ and transition to state $j$. $\pi^\star(i)$ corresponds to the action $u$ that minimizes the r.h.s. in the above equation for evaluation of $V(i)$.

\emph{Relative Value Iteration Algorithm} The fixed point equation~\eqref{fp} can be solved by carrying out the following fixed point iterations. Denote the estimate of the value function at iteration $k$ by $V_k$. Then, the value function is updated according to
\begin{align}\label{rvi1}
\tilde{V}_{k+1}(i) = \min_{u\in \{0,1\}} \sum_{j\in [0,B]} e^{C(i,j,u)}p(j|i,u)V_k(j), i\in[0,B].
\end{align} 
Thereafter normalize the iterates so that,
\begin{align}\label{rvi2}
V_{k+1} (i)= \frac{\tilde{V}_{k+1}(i)}{\tilde{V}_{k+1}(0)},\forall i\in[0,B].
\end{align}
The policy generated at iteration $k$ by the RVI algorithm applies the action that minimizes the r.h.s. of~\eqref{rvi1}. It can be shown that for the RVI iterations, we have that $V_k\to V$, thereby yielding optimal policy~\cite{borkar2002q}. Throughout, for $m\leq n$, we denote by $[m,n]$ the set $\{m,m+1,\ldots,n\}$. 
\section{Structure of the Optimal Policy}\label{sec:struc}
We will show that the optimal policy for the single client scheduling problem is of threshold-type, i.e. it is optimal to carry out packet transmissions only when the queue length $Q(t)$ exceeds a certain threshold $\tau$. The value of threshold $\tau$ depends on the system parameters $p$, and transmission cost $C$. We also show that $\tau$ increases with $C$. 
\begin{definition}
A threshold policy with threshold $\tau$, denoted as $\pi_{\tau}$ schedules packet transmissions at time $t=1,2,\ldots$ only if the queue length $Q(t)\geq \tau$.
\end{definition}
The Relative Value Iteration (RVI) algorithm discussed in the previous section converges, thus yielding optimal policy $\pi^\star$. We will show that at each iteration of the RVI algorithm, the produced policy is of threshold policy. This will prove that the optimal policy is of threshold-type.

Let $V_k$ denote the value function at iteration $k$ of the RVI algorithm. Thus, $V_k(n)$ denotes the relative cost associated with system state being in state $n$.
Let $J_{k+1}(n,1),J_{k+1}(n,0)$ denote the costs associated with applying the actions $U(k+1)=1$ and $U(k+1)=0$ respecively when the system is in state $n$ at stage $k+1$ of the RVI algorithm, i.e., 
\begin{align}\label{qfac}
J_{k+1}(n,0) = 
\begin{cases}
V_k(n+1), n\in [0,B-1]\\
e^{\gamma L}V_k(n), \mbox{ if } n=B,
\end{cases}
\end{align}
\begin{align}
J_{k+1}(n,1) = 
\begin{cases}
e^{\gamma C} V_{k}(n+1), \mbox{ if } n =0,\\
pe^{\gamma (C-R)} V_k(n-1)+ (1-p) e^{\gamma C} V_{k}(n+1), \\
\mbox{ if } n\in [1,B-1],\\
pe^{\gamma (C-R)} V_k(n-1)+ (1-p) e^{\gamma (C+L)} V_{k}(n), \\
\mbox{ if } n = B.
\end{cases}
\end{align}
Let $\partial J_{k+1}(n): =J_{k+1}(n,0)-J_{k+1}(n,1)$ denote the differential between the costs associated with taking the actions $0$ and $1$ if the queue length $Q(k+1)$ at iteration $k+1$ is equal to $n$. The differential $\partial J_{k+1}$ is given as,
\begin{align}\label{df}
\partial J_{k+1}(n)=
\begin{cases}
(1-e^{\gamma C})V_{k}(n+1)\mbox{ if } n=0,\\
V_k(n+1)\left[1-(1-p)e^{\gamma C}\right]\\
 \qquad- pe^{\gamma(C-R)}V_k(n-1),\mbox{ if } n\in[1,B-1],\\
V_k(n)\left[1-(1-p)e^{\gamma C}\right]e^{\gamma L}\\
~\qquad-pe^{\gamma (C-R)}V_k(n-1) \mbox{ if } n=B.
\end{cases}
\end{align}
We clearly have,
\begin{lemma}\label{lemma:basic}
 If the differential $\partial J_{k+1}(n),n\in [0,B]$ is a non-decreasing function of $n$, then the optimal policy produced at iteration $k+1$ by the RVI algorithm~\eqref{rvi1}-~\eqref{rvi2} is of threshold type. 
\end{lemma}
Let us assume that $\partial J_{k+1}$ is non-decreasing in $n$, and try to prove that the function $\partial J_{k+2}$ is non-decreasing in $n$. This result will then imply that the optimal policy produced at iteration $k+2$ is also of threshold type.
\begin{lemma}\label{lemma2}
Let the optimal policy produced by the RVI algorithm at iteration $k+1$ be of threshold-type, with threshold value equal to $\tau$. Then the differential $\partial J_{k+1}$ satisfies
\begin{align}
\partial J_{k+1}(n) &\leq  0, n\in [0,\tau-1],\mbox{ and }\\
\partial J_{k+1}(n) &\geq  0, n\in [\tau,B].
\end{align}
The unscaled value function $\tilde{V}_{k+1}$ produced at iteration $k+1$ is given by,
\begin{align}\label{vf}
\tilde{V}_{k+1}(n) =
\begin{cases}
V_k(n+1), \mbox{ if } n\in \left[0,\tau-1\right],\\
p e^{\gamma (C-R)} V_k(n-1) + (1-p)e^{\gamma C} V_k(n+1), \\
\mbox{ if }n\in\left[\tau,B-1\right],\\
p e^{\gamma (C-R)} V_k(n-1) + (1-p)e^{\gamma (C+L)} V_k(n), \\
\mbox{ if }n= B.
\end{cases}
\end{align}
\end{lemma}
We now show that if the differential $\partial J_{k+1}$ is non-decreasing, then $\partial J_{k+2}$ is also non-decreasing. Since under this assumption, the optimal policy at iteration $k+1$ is of threshold type, we can substitute the value of $V_{k+1}$ derived in Lemma~\ref{lemma2} into the relation for differential~\eqref{df} in order to obtain
\begin{align}\label{eq:arvix}
\partial J_{k+2}(n) = 
\begin{cases}
\left[1-e^{\gamma C}\right]V_k(n+2),n=0\\
\left[1-e^{\gamma C}(1-p)\right]V_k(n+2)\\
-pe^{\gamma (C-R)}V_k(n),n\in [1,\tau-2],\\
\left[pe^{\gamma (C-R)}V_k(n)+e^{\gamma C}(1-p)V_k(n+2) \right]\\
\left[1-e^{\gamma C}(1-p)\right]-pe^{\gamma(C-R)}V_k(n),\\
 n=\tau,\tau-1\\
\left(V_k(n)\left[1-(1-p)e^{\gamma C}\right]\right.\\
\left. - pe^{\gamma(C-R)}V_k(n-2)\right)pe^{\gamma(C-R)}\\
+ \left[\left[1-(1-p)e^{\gamma C}\right]V_k(n+2) \right.\\
\left.- pe^{\gamma (C-R)} V_k(n) \right]\times \\
~\qquad(1-p) e^{\gamma C}, n\in [\tau+1,B-2]\\
\left[pe^{\gamma(C-R)}V_k(n)+(1-p)e^{\gamma(C+L)}V_k(n+1)\right]\\
\times \left[1-(1-p)e^{\gamma C}\right]\\
-pe^{\gamma(C-R)}\left[pe^{\gamma(C-R)}V_k(n-2)\right.\\
\left.+(1-p)e^{\gamma C}V_k(n)\right], n=B-1.
 \end{cases}
 \end{align}
 The expression for $n=B$ is presented in the lemma below.
The above relations can be written more compactly as follows.
\begin{lemma}\label{lemma4}
Assume that the optimal policy at iteration $k+1$ is of threshold type. Then, the differential $\partial J_{k+2}$ is given by
\end{lemma}
\begin{align}\label{df1}
\partial J_{k+2}(n) = 
\begin{cases}
\left[1-e^{\gamma C}(1-p)\right]V_k(n+2)-pe^{\gamma (C)}V_k(n),\\
\mbox{ if }n=0\\
\left[1-e^{\gamma C}(1-p)\right]V_t(n+2)-pe^{\gamma (C-R)}V_k(n),\\
n\in [1,\tau-2],\\
(1-p)e^{\gamma C}\partial J_{k+1}(n+1) ,\mbox{ for }n=\tau-1,\tau\\
\partial J_{k+1}(n-1) pe^{\gamma(C-R)}\\
+ \partial J_{k+1}(n+1)(1-p) e^{\gamma C},  n\in [\tau+1,B-1]\\
 \mbox{ if } n = B.
\end{cases}
\end{align}
while for $n=B$,
\begin{align}\label{df1b}
&\partial J_{k+2}(n) \notag\\
&= [1-(1-p)e^{\gamma C}]e^{\gamma L}\left[p e^{\gamma (C-R)} V_k(n-1) \right.\notag\\
&\qquad\qquad\qquad\qquad\qquad\qquad\left.+ (1-p)e^{\gamma (C+L)} V_k(n)\right] \notag\\
&-pe^{\gamma (C-R)}\left[p e^{\gamma (C-R)} V_k(n-2) + (1-p)e^{\gamma C} V_k(n)\right],
\end{align}

We can now use the expression of $\partial J_{k+2}$ derived in Lemma~\ref{lemma4} in order to show that it is non-decreasing function of $n$.
\begin{lemma}\label{lemma3}
Assume that the differential $\partial J_{k+1}$ at iteration $k+1$ is non-decreasing function of $n$. Then, the differential $\partial J_{k+2}$ at iteration $k+2$ is also non-decreasing in $n$. 
\end{lemma}
\begin{proof}
It follows from Lemma~\ref{lemma4} that for $ n\in [\tau+1,B-1]$, the function $\partial J_{k+2}(n)$ is a linear combination of the functions $\partial J_{k+1}(n-1)$ and $\partial J_{k+1}(n+1)$, both of which are assumed to be non-decreasing functions of $n$. Thus, the claim is true for $n\in [\tau+1,B-1]$. Similar reasoning proves the claim for $n\in [1,\tau-2]$. 

We now verify whether the following two inequalities are true,
\begin{align*}
\partial J_{k+2}(\tau+1)\geq \partial J_{k+2}(\tau)\mbox{ and }
\partial J_{k+2}(\tau)\geq \partial J_{k+2}(\tau-1).
\end{align*}
We note that,
\begin{align*}
\partial J_{k+2}(\tau+1) &= \partial J_{k+1}(\tau) pe^{\gamma(C-R)}\\
&+ \partial J_{k+1}(\tau+2)(1-p) e^{\gamma C}\\
&\geq \partial J_{k+1}(\tau+2)(1-p) e^{\gamma C}\\
&\geq \partial J_{k+1}(\tau+1)(1-p) e^{\gamma C}\\
&=\partial J_{k+2}(\tau),
\end{align*}
where the first inequality follows since the optimal policy at iteration $k+1$ is of threshold type, and from Lemma~\ref{lemma2} we have that $\partial J_{k+1}(\tau)\geq 0$. The second inequality follows from our assumption that $\partial J_{k+1}$ is non-decreasing in $n$, i.e., $\partial J_{k+1}(\tau+2)\geq \partial J_{k+1}(\tau+1)$.

Next, we have,
\begin{align*}
\partial J_{k+2}(\tau) &= (1-p)e^{\gamma C}\partial J_{k+1}(\tau+1)\\
&\geq (1-p)e^{\gamma C}\partial J_{k+1}(\tau)\\
&=\partial J_{k+2}(\tau-1),
\end{align*}
where the inequality follows from our assumption that $\partial J_{k+1}(\tau)$ is non-decreasing.

We now prove $\partial J_{k+2}(0)\leq \partial J_{k+2}(1)$. We substitute the values of $\partial J_{k+2}(0),\partial J_{k+2}(1)$ from Lemma~\ref{lemma4}, so that for $n=0$ we have,
\begin{align*}
\partial J_{k+2}(n) &= \left[1-e^{\gamma C}(1-p)\right]V_k(n+2)-pe^{\gamma (C)}V_k(n)\\
&\leq \left[1-e^{\gamma C}(1-p)\right]V_k(n+2)-pe^{\gamma (C-R)}V_k(n)\\
&\leq \left[1-e^{\gamma C}(1-p)\right]V_k((n+1)+2)\\
&-pe^{\gamma (C-R)}V_k(n+1)\\
&=\partial J_{k+2}(n+1),
\end{align*}
where the first inequality follows since $R>0$, and the second inequality follows since $\partial J_{k+1}$ is assumed to be non-decreasing in $n$.

Finally, we prove $\partial J_{k+2}(B)\geq \partial J_{k+2}(B-1)$. Substituting the values of $\partial J_{k+2}(B)$ from Lemma~\ref{lemma4}, and the vale of $\partial J_{k+2}(B-1)$ from~\eqref{df}, 
the condition $\partial J_{k+2}(B)\geq \partial J_{k+2}(B-1)$ reduces to,
\begin{align*}
&V_{k+1}(B)[1-(1-p)e^{\gamma C}]e^{\gamma L}-pe^{\gamma (C-R)}V_{k+1}(B-1)\\
&\geq V_{k+1}(B)\left[1-(1-p)e^{\gamma C}\right]-pe^{\gamma (C-R)}V_{k+1}(B-2),
\end{align*}
or equivalently
\begin{align*}
&V_{k+1}(B)\left[1-(1-p)e^{\gamma C}\right]\left[e^{\gamma L}-1\right]\\
&\geq pe^{\gamma (C-R)}\left[V_{k+1}(B-1)-V_{k+1}(B-2)\right]
\end{align*}
This concludes the proof.

\end{proof}
\begin{theorem}[Optimality of Threshold Policy]\label{th1}
For the single client risk-sensitive scheduling problem of minimizing the infinite-horizon cost~\eqref{rsmdp}, a threshold policy is optimal.
\end{theorem}
\begin{proof}
We will use induction on the iteration number $k$ of the RVI algorithm in order to prove the theorem. For the RVI algorithm, let us initialize the $V_0(n)=1,\forall n\in[0,B]$. It then follows that,
\begin{align}\label{df3}
\partial J_{1}(n)=
\begin{cases}
(1-e^{\gamma C})\mbox{ if } n=0,\\
\left[1-(1-p)e^{\gamma C}\right]- pe^{\gamma(C-R)},\mbox{ if } n\in[1,B-1],\\
\left[1-(1-p)e^{\gamma C}\right]e^{\gamma L}-pe^{\gamma (C-R)} \mbox{ if } n=B.
\end{cases}
\end{align}
It is easily verified that $\partial J_1$ is non-decreasing in $n$. Thus, it now follows from Lemma~\ref{lemma3}, that at each iteration $k$ of the RVI algorithm, the function $\partial J_k$ is non-decreasing in $n$. Thus, from Lemma~\ref{lemma:basic} we have that the policy produced by the RVI algorithm at each iteration $k$ is of threshold type. Since the RVI algorithm converges to the optimal policy, the optimal policy is also of threshold type.  
\end{proof}
Next, we show that for the optimal policy $\pi^\star$, the threshold denoted as $\tau^\star$ increases with the transmission cost $C$. The following condition ensures that the threshold of the policy produced by the RVI algorithm at stage $k+1$ is an increasing function of transmission cost $C$.
\begin{condition}[Monotonicity]\label{cond:1}
If $C_1,C_2>0$ are such that $C_1>C_2$, then $\partial J^{C_1}_{k+1}(m) \leq \partial J^{C_2}_{k+1}(m)$ for each $m\in [0,B]$. 
\end{condition}
\begin{lemma}\label{monoc}
Assume that the Condition~\ref{cond:1} is true for the RVI algorithm at iteration $k$. Then, the Condition~\ref{cond:1} also holds true at iteration $k+1$ of the RVI algorithm, and hence for the policy produced at iteration $k+1$, the threshold value is an increasing function of the transmission cost $C$.
\end{lemma}
\begin{proof}
In the ensuing discussion, we let $V_{k,1}$ denote the value function associated with $k$-th iteration of RVI applied to the risk-sensitive control problem~\eqref{rsmdp} with transmission cost set at $C_1$, while $\partial J^{C_1}_k$ will denote the corresponding cost differential. Similarly for $V_{k,2},\partial J^{C_2}_k$.

In order to prove the claim, we need to show that $ \partial J^{C}_{k+2}(n)$ is increasing function of $C$ for each $n\in[0,B]$.
For an $n\in [\tau+1,B-1]$, and $C_1>C_2>0$ we have,
\begin{align*}
&\partial J^{C_1}_{k+2}(n) \\
&=(1-p)e^{\gamma C_1}\partial J^{C_1}_{k+1}(n+1) + pe^{\gamma(C_1-R)} \partial J^{C_1}_{k+1}(n-1)\\
&\leq (1-p)e^{\gamma C_2}\partial J^{C_2}_{k+1}(n+1) + pe^{\gamma(C_2-R)} \partial J^{C_2}_{k+1}(n-1)\\
&=J^{C_2}_{k+2}(n),
\end{align*}
where the equalities follow from the relation~\eqref{df1} and the inequality follows from our assumption that the  Condition~\ref{cond:1} is satisfied at iteration $k+1$ of the RVI algorithm.

Next, we prove the claim for $n\in [1,\tau-1]$. For $C_1>C_2>0$ and $n\in [1,\tau-2]$ we have
\begin{align*}
&\partial J^{C_1}_{k+2}(n) \\
&=\left[1-e^{\gamma C_1}(1-p)\right]V_{k,1}(n+2)-pe^{\gamma (C_1-R)}V_{k,1}(n)\\
&=\partial J^{C_1}_{k}(n+1)\\
&\leq \partial J^{C_2}_{k}(n+1)\\
&=\partial J^{C_2}_{k+2}(n),
\end{align*}
where the inequality results from Condition~\ref{cond:1}.

Now we prove the desired condition for $n=0$. It follows from~\eqref{df} that the condition $\partial J^{C_1}_{k+2}(0)\leq \partial J^{C_2}_{k+2}(0)$ reduces to $(1-e^{\gamma C})V_{k+1,1}(1)\leq (1-e^{\gamma C_2})V_{k+1,2}(1)$. Since $C_1,C_2>0$ the condition is equivalent to $V_{k+1,1}(1)\geq V_{k+1,2}(1)$. Fix a time horizon $T>0$, and a scheduling policy $\pi$, and consider the operation of two systems under the application of the policy $\pi$. The transition probabilities of the two controlled Markovian systems are taken to be the same, but their transmission costs are set at $C_1$ and $C_2$. Construct their sample paths on the same probability space. It now follows from stochastic coupling~\cite{thorisson1995coupling}, that the sample path cost $\sum_{t=1}^{T} CU(t) - R(t)+L(t)$, or equivalently the cost $e^{\gamma \sum_{t=1}^{T} CU(t) - R(t)+L(t)}$ incurred by the system with cost set at $C_1$ is greater than or equal to the system with cost equal to $C_2$. Hence it follows that $V_{k+1,1}(1)\geq V_{k+1,2}(1)$.

For $n=\tau,\tau-1$, the differential $ \partial J^{C_1}_{k+2}(n)$ yields us
\begin{align*}
&\partial J^{C_1}_{k+2}(n)\\
&=(1-p)e^{\gamma C_1}\partial J^{C_1}_{k+1}(n+1)\\
&\leq (1-p)e^{\gamma C_1}\partial J^{C_2}_{k+1}(n+1)\\
&=\partial J^{C_2}_{k+2}(n),
\end{align*}
where the equality follows from the relation~\eqref{df1}, and the inequality results from Condition~\ref{cond:1}. 
\end{proof}

\begin{theorem}\label{th:2}
Consider the problem of designing a scheduling policy that makes decisions regarding packet transmissions in order to minimize the infinite horizon risk-sensitive cost~\eqref{rsmdp}. For $C_1>C_2>0$, let $\tau_{C_1}$ and $\tau_{C_2}$ denote the threshold values of the optimal policies when transmission costs are set at $C_1$ and $C_2$ respectively. We then have $\tau_{C_1}\geq \tau_{C_2}$.
\end{theorem}
\begin{proof}
Consider the optimal risk-sensitive control problem~\eqref{rsmdp} with transmission cost set at $C$, and initialize $V_{0}(n)=1,\forall n\in[0,B]$. We then have 
\begin{align}
\partial J_{1}(n)=
\begin{cases}
(1-e^{\gamma C})\mbox{ if } n=0,\\
\left[1-(1-p)e^{\gamma C}\right]- pe^{\gamma(C-R)},\mbox{ if } n\in[1,B-1],\\
\left[1-(1-p)e^{\gamma C}\right]e^{\gamma L}-pe^{\gamma (C-R)} \mbox{ if } n=B.
\end{cases}
\end{align}
It is easily verified that $\partial J_{1}(n)$ is non-increasing function of $C$, and hence Condition~\ref{cond:1} holds true at iteration $k=1$ of the RVI algorithm. 

The result now follows by using induction on iteration number $k$ in conjunction with Lemma~\ref{monoc}.
\end{proof}
\section{Computing the Optimal Policy}\label{sec:opt}
Having derived the structure of the optimal policy, we would like to compute the value of threshold $\tau$ corresponding to the optimal policy. In view of Theorem~\ref{th:2}, we will derive the set of values of transmission cost $C$ such that the policy $\pi_{\tau}$ is optimal when the transmission cost is set at $C$.

 It follows from the optimality conditions~\eqref{fp} that the following set of $B+1$ equations need to be solved in order to derive the performance of $\pi_{\tau}$.
\begin{align}
\alpha V(0) & = V(1),\label{recur} \\
\alpha V(i) &= V(i+1)~\forall i\in\left[1,\tau-1\right],\label{recur1}\\
\alpha V(i) &= \exp(\gamma C ) \left(p\exp(-\gamma R)V(i-1)\right.\notag\\
&\qquad\qquad\left.+(1-p)V(i+1)\right),~i\in\left[\tau,B-1\right]\label{recur2}\\
\alpha V(B) &= \exp(\gamma C) \left(p \exp(-\gamma R)V(B-1) \right.\notag\\
&\qquad\qquad\left.+ (1-p)e^{\gamma L}V(B) \right)\label{recur3}
\end{align}
where $\alpha$ is the exponential of the infinite horizon risk-sensitive cost, and $V(i)$ is the relative cost associated with the system starting in state $i$. We now solve the set of equations~\eqref{recur}-\eqref{recur3}. Clearly,
\begin{align}\label{v1}
V(i) = \alpha^i, i\in\left[0,\tau\right].
\end{align}
The characteristic equation corresponding to recursive relations~\eqref{recur2} is given by,
\begin{align*}
(1-p)\lambda^2 -\alpha e^{-\gamma C} \lambda + pe^{-\gamma R}=0,
\end{align*}
whose solutions are given by,
\begin{align*}
\lambda_1,\lambda_2 &= \frac{\alpha e^{-\gamma C} +,-\sqrt{(\alpha e^{-\gamma C} )^2-4p(1-p)e^{-\gamma R}}}{2(1-p)}.
\end{align*}
Thus, for $i\in [0,B-\tau]$, we have,
 \begin{align}\label{v2}
V(\tau+i) = K_1 \lambda_1^{i+1}+K_2\lambda_2^{i+1}, i=1,2,\ldots,B-\tau,
\end{align}
The initial conditions for the recursions~\eqref{recur2} are determined by the evaluation of $w(\tau-1)$ and $w(\tau)$, i.e.,
\begin{align*}
K_1 + K_2 &= V(\tau-1) =\alpha^{\tau-1}, \\
K_1\lambda_1 + K_2 \lambda_2 &= V(\tau) = \alpha^\tau.
\end{align*}
Solving for $K_1,K_2$ in terms of $w(\tau-1),w(\tau)$ we get,
\begin{align}
K_1 &= \frac{\alpha^{\tau-1}(\alpha-\lambda_2)(1-p)}{\delta},\label{k1}\\
K_2 &= \frac{\alpha^{\tau-1}(\lambda_1-\alpha)(1-p)}{\delta}.\label{k2}
\end{align}
The average cost $\alpha$ can be obtained by utilizing the boundary condition at $i=B$, i.e, the equation~\eqref{recur3},
\begin{align}\label{eq:alpha}
&\left[\alpha-(1-p)e^{\gamma(C+L)}\right] \left(K_1 \lambda_1^{B-\tau+1}+K_2\lambda_2^{B-\tau+1}\right)\notag\\
&=e^{\gamma C}pe^{-\gamma R} \left(K_1 \lambda_1^{B-\tau}+K_2\lambda_2^{B-\tau}\right)
\end{align}
We now find the values of transmission cost $C$, for which $\pi_{\tau}$ is optimal for the risk-sensitive scheduling problem with cost set at $C$.

Let $\partial J(n)$ denote the limit value of $\partial J_k$ obtained upon convergence of the RVI algorithm. It follows from Lemma~\ref{monoc} and the analysis of Theorem~\ref{th:2} that $\partial J(n)$ is a non-increasing function of the cost $C$ for each value of the system state $n$. Hence, the necessary and sufficient condition for $\pi_{\tau}$ to be optimal are 
\begin{align}
\partial J(\tau-1) &\leq 0, \mbox{ and },\label{cond:thrh2}\\
\partial J(\tau) &\geq 0, \label{cond:thrh1}
\end{align}

Since the function $\partial J(n)$ was shown to be non-increasing in $C$ for each $n$, it follows that the set of costs $C$ which satisfy the inequality~\eqref{cond:thrh2} is of the form $[C_l,\infty)$, while the solution set of inequality~\eqref{cond:thrh1} is of the form $[0,C_u]$, for some suitable values of $C_l,C_u\geq 0$. Since for a fixed cost $C$, the function $\partial J(n)$ is non-decreasing in $n$, it follows that $C_u\geq C_l$, and hence $\pi_\tau$ is optimal when $C\in [C_l,C_u]$.
\begin{theorem}\label{th:3}
Consider the class comprising of optimal risk-sensitive control problems parametrized by transmission cost $C$, in which for each individual risk-sensitive MDP the cost incurred is given by~\eqref{rsmdp}. Then, the threshold policy $\pi_{\tau}$ is optimal for risk-sensitive MDPs for which the cost $C\in[C_l,C_u]$, where $C_l,C_u$ can be obtained by solving the equations~\eqref{cond:thrh2},~\eqref{cond:thrh1}. 
\end{theorem}
Let us now re-write the equation~\eqref{cond:thrh2},~\eqref{cond:thrh1} in terms of parameters $\lambda_1,\lambda_2,p,\gamma$. 
The quantities $C_l,C_u$ can be obtained by substituting the values of $\partial J$ into the above conditions. 

Similar to the relations~\eqref{df}, the steady-state differentials $\partial J(n)$ are calculated as,
\begin{align}\label{df2}
\partial J(n)=
\begin{cases}
w(n+1) - e^{\gamma C} \left[pV(n)+(1-p)V(n+1)\right],\\
\mbox{ if }n=0,\\
w(n+1)\left[1-(1-p)e^{\gamma C}\right] - pe^{\gamma(C-R)}V(n-1),\\
\mbox{ if } n\in[1,B-1],\\
V(n)-e^{\gamma C} \left[pe^{-\gamma R}V(n-1)+(1-p)e^{\gamma L}V(n)\right],\\
\mbox{ if } n=B.
\end{cases}
\end{align}
The value function $V$ can be substituted from~\eqref{v1} and~\eqref{v2} into the above relation, and thereafter the resulting $\partial J$ can be substituted into the inequalities~\eqref{cond:thrh2},~\eqref{cond:thrh1} in order to yield the desired equations. In summary, the solution of two equations~\eqref{cond:thrh2},~\eqref{cond:thrh1} solves a set of risk-sensitive optimal control problems parameterized by the transmission cost $C$.


\section{Conclusion and Future Works}\label{conclu}
We have derived the optimal risk-sensitive scheduling policy for a single client being served by a wireless channel. The otimal policy was shown to have a threshold structure, and hence is easily implementable. Furthermore we showed that the threshold increases with packet transmission cost, and hence the policy with threshold set at $\tau$ is optimal when the transmission cost lies within the interval $[C_l,C_u]$. The quantities $C_l,C_u$ can be derived by solving two equations. We plan to extend the analysis to the case where multiple clients share a single wireless channel, and the AP has to prioritize the clients for packet transmissions, based on their queue lengths. We would also like to consider the scenario where the transmitter can choose to transmit from amongst various power levels, where a transmission involving higher power having a higher service rate. 

 \bibliographystyle{IEEEtran}
\bibliography{combinedbib}

\begin{thebibliography}{10}
\providecommand{\url}[1]{#1}
\csname url@samestyle\endcsname
\providecommand{\newblock}{\relax}
\providecommand{\bibinfo}[2]{#2}
\providecommand{\BIBentrySTDinterwordspacing}{\spaceskip=0pt\relax}
\providecommand{\BIBentryALTinterwordstretchfactor}{4}
\providecommand{\BIBentryALTinterwordspacing}{\spaceskip=\fontdimen2\font plus
\BIBentryALTinterwordstretchfactor\fontdimen3\font minus
  \fontdimen4\font\relax}
\providecommand{\BIBforeignlanguage}[2]{{%
\expandafter\ifx\csname l@#1\endcsname\relax
\typeout{** WARNING: IEEEtran.bst: No hyphenation pattern has been}%
\typeout{** loaded for the language `#1'. Using the pattern for}%
\typeout{** the default language instead.}%
\else
\language=\csname l@#1\endcsname
\fi
#2}}
\providecommand{\BIBdecl}{\relax}
\BIBdecl

\bibitem{cardenas2008research}
A.~A. C{\'a}rdenas, S.~Amin, and S.~Sastry, ``Research challenges for the
  security of control systems.'' in \emph{HotSec}, 2008.

\bibitem{coraluppi1999risk}
S.~P. Coraluppi and S.~I. Marcus, ``Risk-sensitive and minimax control of
  discrete-time, finite-state markov decision processes,'' \emph{Automatica},
  vol.~35, no.~2, pp. 301--309, 1999.

\bibitem{James1992}
M.~James, ``Asymptotic analysis of nonlinear stochastic risk-sensitive control
  and differential games,'' \emph{Mathematics of Control, Signals and Systems},
  vol.~5, no.~4, pp. 401--417, 1992.

\bibitem{Jacobson1973}
D.~Jacobson, ``Optimal stochastic linear systems with exponential performance
  criteria and their relation to deterministic differential games,''
  \emph{Automatic Control, IEEE Transactions on}, vol.~18, no.~2, pp. 124--131,
  Apr 1973.

\bibitem{bensoussan1985optimal}
A.~Bensoussan and J.~Van~Schuppen, ``Optimal control of partially observable
  stochastic systems with an exponential-of-integral performance index,''
  \emph{SIAM Journal on Control and Optimization}, vol.~23, no.~4, pp.
  599--613, 1985.

\bibitem{james1994risk}
M.~R. James, J.~S. Baras, and R.~J. Elliott, ``Risk-sensitive control and
  dynamic games for partially observed discrete-time nonlinear systems,''
  \emph{IEEE transactions on automatic control}, vol.~39, no.~4, pp. 780--792,
  1994.

\bibitem{Marcus1997}
S.~I. Marcus, E.~Fernandez-Gaucherand, D.~Hernandez-Hernandez, S.~Coraluppi,
  and P.~Fard, ``Risk sensitive {M}arkov decision processes,'' in \emph{Systems
  and control in the 21st century}, 1997.

\bibitem{Fleming1997}
W.~Fleming and D.~Hernandez-Hernandez, ``Risk sensitive control of finite state
  machines on an infinite horizon. i,'' in \emph{Decision and Control, 1997.,
  Proceedings of the 36th IEEE Conference on}, vol.~4, Dec 1997, pp. 3407--3412
  vol.4.

\bibitem{fleming1995risk}
W.~H. Fleming and W.~M. McEneaney, ``Risk-sensitive control on an infinite time
  horizon,'' \emph{SIAM Journal on Control and Optimization}, vol.~33, no.~6,
  pp. 1881--1915, 1995.

\bibitem{kumar1981optimal}
P.~Kumar and J.~Van~Schuppen, ``On the optimal control of stochastic systems
  with an exponential-of-integral performance index,'' \emph{Journal of
  mathematical analysis and applications}, vol.~80, no.~2, pp. 312--332, 1981.

\bibitem{rahul}
R.~Singh, I.-H. Hou, and P.~Kumar, ``Fluctuation analysis of debt based
  policies for wireless networks with hard delay constraints,'' in \emph{IEEE
  INFOCOM, 2014 Proceedings}, April 2014, pp. 2400--2408.

\bibitem{rahul1}
{Rahul Singh, I-Hong Hou and P.R. Kumar}, ``Pathwise performance of debt based
  policies for wireless networks with hard delay constraints,'' in
  \emph{Decision and Control (CDC), 2013 IEEE 52nd Annual Conference on}, Dec
  2013, pp. 7838--7843.

\bibitem{guosingh}
\BIBentryALTinterwordspacing
X.~Guo, R.~Singh, P.~Kumar, and Z.~Niu, ``A high reliability asymptotic
  approach for packet inter-delivery time optimization in cyber-physical
  systems,'' in \emph{Proceedings of the 16th ACM International Symposium on
  Mobile Ad Hoc Networking and Computing}, ser. MobiHoc '15.\hskip 1em plus
  0.5em minus 0.4em\relax New York, NY, USA: ACM, 2015, pp. 197--206. [Online].
  Available: \url{http://doi.acm.org/10.1145/2746285.2746305}
\BIBentrySTDinterwordspacing

\bibitem{singh_stolyar_info}
R.~Singh and A.~Stolyar, ``Maxweight scheduling: "smoothness" of the service
  process,'' in \emph{IEEE INFOCOM 2016 - The 35th Annual IEEE International
  Conference on Computer Communications}, April 2016, pp. 1--9.

\bibitem{rs}
------, ``Maxweight scheduling: Asymptotic behavior of unscaled
  queue-differentials in heavy traffic,'' in \emph{Proceedings of the 2015 ACM
  SIGMETRICS International Conference on Measurement and Modeling of Computer
  Systems}, ser. SIGMETRICS '15.\hskip 1em plus 0.5em minus 0.4em\relax New
  York, NY, USA: ACM, 2015, pp. 431--432.

\bibitem{singh2017optimal}
R.~Singh and E.~Modiano, ``Optimal routing for delay-sensitive traffic in
  overlay networks,'' \emph{arXiv preprint arXiv:1703.07419}, 2017.

\bibitem{singh2016throughput}
R.~Singh and P.~Kumar, ``Throughput optimal decentralized scheduling of
  multi-hop networks with end-to-end deadline constraints: Unreliable links,''
  \emph{arXiv preprint arXiv:1606.01608}, 2016.

\bibitem{Howard1972}
R.~A. Howard and J.~E. Matheson, ``Risk-sensitive {M}arkov decision
  processes,'' \emph{Management Science}, vol.~18, no.~7, pp. pp. 356--369,
  1972.

\bibitem{whittle1990risk}
P.~Whittle, ``Risk-sensitive optimal control,'' 1990.

\bibitem{bompard2009risk}
E.~Bompard, C.~Gao, R.~Napoli, A.~Russo, M.~Masera, and A.~Stefanini, ``Risk
  assessment of malicious attacks against power systems,'' \emph{IEEE
  Transactions on Systems, Man, and Cybernetics-Part A: Systems and Humans},
  vol.~39, no.~5, pp. 1074--1085, 2009.

\bibitem{ericsson2007toward}
G.~N. Ericsson, ``Toward a framework for managing information security for an
  electric power utility?cigr{\'e} experiences,'' \emph{IEEE transactions on
  power delivery}, vol.~22, no.~3, pp. 1461--1469, 2007.

\bibitem{amin2009safe}
S.~Amin, A.~A. C{\'a}rdenas, and S.~S. Sastry, ``Safe and secure networked
  control systems under denial-of-service attacks,'' in \emph{International
  Workshop on Hybrid Systems: Computation and Control}.\hskip 1em plus 0.5em
  minus 0.4em\relax Springer, 2009, pp. 31--45.

\bibitem{befekadu2011risk}
G.~K. Befekadu, V.~Gupta, and P.~J. Antsaklis, ``Risk-sensitive control under a
  class of denial-of-service attack models,'' in \emph{American Control
  Conference (ACC), 2011}.\hskip 1em plus 0.5em minus 0.4em\relax IEEE, 2011,
  pp. 643--648.

\bibitem{singh2015index}
R.~Singh, X.~Guo, and P.~R. Kumar, ``Index policies for optimal mean-variance
  trade-off of inter-delivery times in real-time sensor networks,'' in
  \emph{Computer Communications (INFOCOM), 2015 IEEE Conference on}.\hskip 1em
  plus 0.5em minus 0.4em\relax IEEE, 2015, pp. 505--512.

\bibitem{1102957}
{Z. Rosberg, P. Varaiya and J. Walrand}, ``Optimal control of service in tandem
  queues,'' \emph{IEEE Transactions on Automatic Control}, vol.~27, no.~3, pp.
  600--610, Jun 1982.

\bibitem{buyukkoc1985c}
C.~Buyukkoc, P.~Variaya, and J.~Walrand, ``c mu rule revisited.'' \emph{Adv.
  Appl. Prob.}, vol.~17, no.~1, pp. 237--238, 1985.

\bibitem{lin1984optimal}
W.~Lin and P.~Kumar, ``Optimal control of a queueing system with two
  heterogeneous servers,'' \emph{IEEE Transactions on Automatic control},
  vol.~29, no.~8, pp. 696--703, 1984.

\bibitem{stidham1989monotonic}
S.~Stidham~Jr and R.~R. Weber, ``Monotonic and insensitive optimal policies for
  control of queues with undiscounted costs,'' \emph{Operations Research},
  vol.~37, no.~4, pp. 611--625, 1989.

\bibitem{george2001dynamic}
J.~M. George and J.~M. Harrison, ``Dynamic control of a queue with adjustable
  service rate,'' \emph{Operations Research}, vol.~49, no.~5, pp. 720--731,
  2001.

\bibitem{harisondiscount}
J.~M. Harrison, ``Dynamic scheduling of a multiclass queue: Discount
  optimality,'' \emph{Operations Research}, vol.~23, no.~2, pp. 270--282, 1975.

\bibitem{tassi1}
L.~Tassiulas and A.~Ephremides, ``Stability properties of constrained queueing
  systems and scheduling policies for maximum throughput in multihop radio
  networks,'' \emph{IEEE Transactions on Automatic Control}, vol.~37, no.~12,
  pp. 1936--1948, Dec 1992.

\bibitem{sennott2009stochastic}
L.~I. Sennott, \emph{Stochastic dynamic programming and the control of queueing
  systems}.\hskip 1em plus 0.5em minus 0.4em\relax John Wiley \& Sons, 2009,
  vol. 504.

\bibitem{borkar2002q}
V.~S. Borkar, ``Q-learning for risk-sensitive control,'' \emph{Mathematics of
  operations research}, vol.~27, no.~2, pp. 294--311, 2002.

\bibitem{Rojas1998}
A.~Brau-Rojas, R.~Cavazos-Cadena, and E.~Fernandez-Gaucherand, ``Controlled
  {M}arkov chains with risk-sensitive criteria: some (counter) examples,'' in
  \emph{Decision and Control, 1998. Proceedings of the 37th IEEE Conference
  on}, vol.~2, Dec 1998, pp. 1853--1858 vol.2.

\bibitem{thorisson1995coupling}
H.~Thorisson, ``Coupling methods in probability theory,'' \emph{Scandinavian
  journal of statistics}, pp. 159--182, 1995.

\end{thebibliography}
\end{document}